\documentclass[12pt]{scrartcl}

\usepackage{amssymb}
\usepackage{amsthm}
\usepackage{amsmath}
\usepackage{bbm}
\usepackage{clrscode}
\usepackage{graphicx}

\theoremstyle{plain} 
\newtheorem{thm}{Theorem}
\newtheorem{lem}{Lemma}
\newtheorem{cor}{Corollary}
\theoremstyle{definition}
\newtheorem{defn}{Definition}

\newcommand{\N}{\mathbbmss{N}}
\newcommand{\Z}{\mathbbmss{Z}}

\renewcommand{\L}{\mathcal{L}}
\newcommand{\OO}{\mathcal{O}}
\newcommand{\powset}[1]{2^{#1}}
\newcommand{\len}[1]{{\vert #1 \vert}}
\newcommand{\substr}[3]{#1[#2,#3]}
\newcommand{\chr}[2]{#1[#2]}
\newcommand{\mboxspace}[1]{\mbox{\hspace{1em} #1 \hspace{1em}}}
\newcommand{\matches}{\vartriangleleft}
\newcommand{\suffix}[2]{#1[#2..]}
\newcommand{\prefix}[2]{#1[..#2]}

\DeclareMathOperator*{\closure}{close_\sqsubseteq}

\title{Construction of minimal DFAs from biological motifs}
\author{Tobias Marschall\\
\small Bioinformatics for High-Throughput Technologies,\\[-0.8ex]
\small Computer Science XI, TU Dortmund, Germany\\[-0.8ex]
\small \texttt{tobias.marschall@tu-dortmund.de}
}
\date{}

\begin{document}
\maketitle

\begin{abstract}
\noindent
Deterministic finite automata (DFAs) are constructed for various purposes in computational biology. Little attention, however, has been given to the efficient construction of minimal DFAs. In this article, we define \emph{simple} non-deterministic finite automata (NFAs) and prove that the standard subset construction transforms NFAs of this type into \emph{minimal} DFAs. Furthermore, we show how simple NFAs can be constructed from two types of patterns popular in bioinformatics, namely (sets of) generalized strings and (generalized) strings with a Hamming neighborhood.
\end{abstract}


\section{Introduction}
\label{sec:introduction}
Deterministic and non-deterministic finite automata belong to the curriculum of every theoretical computer scientist. It is well known that, given a non-deterministic finite automaton (NFA), we can construct a deterministic finite automaton (DFA) recognizing the same language by employing the classical subset construction; each state in the resulting DFA corresponds to a set of NFA states. The details can be found in many textbooks on the topic, for example in~\cite{Hopcroft1979,Kozen1999Automata,Sipser2005Introduction}. If $Q$ is an NFA's finite state space, then there are $2^{|Q|}$ subsets and hence the same number of DFA states. In most cases, many of these states turn out to be inaccessible from the start state and can be discarded. In practice, we can use a construction scheme that only generates the accessible states by performing a breadth-first search on the state space~\cite{Navarro2002Flexible}. For each DFA, there exists a unique (up to isomorphism) minimal DFA that accepts the same language~\cite{Kozen1999Automata}. Following the subset construction, we may thus want to minimize the resulting DFA, for example by using Hopcroft's algorithm~\cite{Hop71,Knuutila2001Redescribing}.

In computational biology, the processing of sequences plays a prominent role. Sequences of nucleotides (DNA or RNA) and amino acids (proteins) are key players in the biology of cells. Recurring elements in such sequences, called \emph{patterns} or \emph{motifs}, can often be associated with biological function~\cite{HulBaiBui06,AlbWynEng04}. Three important problem fields in connection with motifs are those of \emph{motif search}~\cite{Navarro2002Flexible}, \emph{motif statistics}~\cite{ReiSchWat00,Reg00,NicSalFla02,LlaBetKni08,DBLP:conf/cpm/MarschallR08} and \emph{motif discovery}~\cite{TomLiBai05,LiTom06,SanDra06,MarRah09}. Not surprisingly, in many algorithms in these fields, motifs are transformed into deterministic automata recognizing all possible instances of the motif. Motivated by this observation, we explore the construction of minimal DFAs for two common motif classes, namely (sets of) generalized strings and consensus strings with a Hamming neighborhood. Ultimately, the goal is to find algorithms whose runtime depends linearly on the number of states of the minimal DFA (which would be optimal). Although automata theory has been subject to extensive research for decades, not much attention has been given to this particular topic. Recently in 2008, van~Glabbeek and Ploeger~\cite{GlaPlo08} addressed the problem of determinization and integrated minimization. In Section~\ref{sec:alt_proof}, we discuss the connections between their work and this article.

\paragraph{Our contributions} We identify a class of NFAs that directly result in minimal DFAs when subjected to the classical subset construction. Although the concept is quite simple and seemingly restrictive, we show that it is strong enough to cover many patterns found in computational biology. To this end, we give construction schemes to transform (sets of) generalized strings and consensus strings with a Hamming neighborhood into NFAs which exhibit this property.

The article is organized as follows. First, we establish notation by briefly re-stating textbook definitions of automata in Section~\ref{sec:notation}. Then, in Section~\ref{sec:simple_nfa}, we introduce the concept of \emph{simple NFA} and show that applying the subset construction to a simple NFA directly yields a minimal DFA. The theory is put to work in Sections~\ref{sec:app_to_gen_strings} and~\ref{sec:app_to_consensus_strings}, where we discuss the construction of minimal DFAs from generalized strings and consensus strings, respectively.

\section{Notation and Basic Definitions}\label{sec:notation}
Let $\Sigma$ be a finite alphabet and let $\Sigma^k$ be the set of all \emph{strings} of length~$k$. Then, the set of all finite strings $\bigcup_{i=0}^\infty\Sigma^i$ is denoted $\Sigma^\ast$ and $\bigcup_{i=1}^\infty\Sigma^i$ is denoted $\Sigma^+$. For a string $s\in\Sigma^\ast$, its length is written $\len{s}$, and $s_1 s_2$ denotes the concatenation of $s_1$ and $s_2$. The only string $\varepsilon\in\Sigma^\ast$ such that $\len{\varepsilon}=0$ is called \emph{empty string}. By $\chr{s}{i}$, we refer to the $i$-th character of $s$, i.e.\ $s=\chr{s}{1}\chr{s}{2}\ldots\chr{s}{\len{s}}$. Furthermore, $\substr{s}{i}{j}:=\chr{s}{i}\chr{s}{i+1}\ldots\chr{s}{j}$ refers to a substring of $s$.
If $i>j$, we define $\substr{s}{i}{j}:=\varepsilon$. Prefixes and suffixes of $s$ are written $\prefix{s}{i}:=\substr{s}{1}{i}$ and $\suffix{s}{i}:=\substr{s}{i}{\len{s}}$, respectively.

We can extend the notion of a string in a natural way by allowing a \emph{generalized string} to be a sequence of sets of characters:
\begin{defn}[Generalized string]
Given an alphabet $\Sigma$, we call the set $\mathcal{G}_\Sigma:=\powset{\Sigma}\setminus\lbrace\emptyset\rbrace$ \emph{generalized alphabet over $\Sigma$} and a string over $\mathcal{G}_\Sigma$ \emph{generalized string}. By $\mathcal{G}_\Sigma^k$ and $\mathcal{G}_\Sigma^*$, we refer to the set of all generalized strings of length~$k$ and the set of all generalized strings of finite length, respectively. We say a string $s\in\Sigma^\ast$ \emph{matches} the generalized string $g\in\mathcal{G}_\Sigma^\ast$, written $s\matches g$, if $\len{s}=\len{g}$ and $\chr{s}{i}\in \chr{g}{i}$ for $1\leq i\leq\len{g}$.
\end{defn}

We write $\mathcal{G}$ instead of $\mathcal{G}_\Sigma$ if the used alphabet is clear from the context. Note that every string $s\in\Sigma$ can be translated into the generalized string $\lbrace \chr{s}{1}\rbrace\lbrace \chr{s}{2}\rbrace\ldots\lbrace \chr{s}{\len{s}}\rbrace$. In this sense, strings can be seen as special cases of generalized strings. Let us now proceed to the classical definitions of automata.

\begin{defn}[Deterministic finite automaton (DFA)]\label{def:dfa}
A \emph{deterministic finite automaton} is a tuple $(Q, \Sigma, \delta, q_\alpha, F)$, where
$Q$ is a finite set of states, $\Sigma$ is a finite alphabet, $\delta: Q\times\Sigma\rightarrow Q$ is a \emph{transition function}, $q_\alpha\in Q$ is the \emph{start state}, and $F\subset Q$ is the set of \emph{accepting states}.
\end{defn}

\begin{defn}[Non-deterministic finite automaton (NFA)]\label{def:nfa}
A \emph{non-deterministic finite automaton} is a tuple $(Q, \Sigma, \Delta, Q_\alpha, F)$, where $Q$, $\Sigma$ and $F$ are defined as for the DFA above, $\Delta: Q\times\Sigma\rightarrow \powset{Q}$ is the \emph{non-deterministic transition function} and $Q_\alpha\subset Q$ is a \emph{set of start states}.
\end{defn}

Note that using a set $Q_\alpha$ instead of only one start state is a notational convenience rather than a conceptual change: we can always transform the automaton to have only one start state by adding the start state $q_\alpha$ and defining its outgoing transitions by $(q_\alpha,\sigma)\mapsto\bigcup_{q\in Q_\alpha}\Delta(q,\sigma)$.

Another convenience is the extension of a DFA's transition function to strings (instead of single characters):
\begin{align*}
\hat{\delta}: Q\times\Sigma^\ast &\rightarrow Q \\
(q,s) &\mapsto
\begin{cases}
q & \mbox{if }s=\varepsilon\mbox{\,,}\\
\hat{\delta}\big(\delta(q,\chr{s}{1}),\suffix{s}{2}\big) & \mbox{otherwise\,.}
\end{cases}
\end{align*}
Analogously, the transition function $\Delta$ of an NFA can be extended to $\hat{\Delta}$. Furthermore, we define $\mathcal{L}(q):=\lbrace s\in\Sigma^\ast\,\vert\,\hat{\Delta}(q,s)\cap F\neq\emptyset\rbrace$ and call it \emph{language of state~$q$}. The language of a set of states~$Q'$ is defined as $\mathcal{L}(Q'):=\bigcup_{q'\in Q'}\mathcal{L}(q')$. Following~\cite{Berstel1985}, we call a state $q\in Q$ \emph{accessible}, if there exist a string $s\in\Sigma^*$ and a start state $q_\alpha\in Q_\alpha$ such that $\hat{\Delta}(q_\alpha,s)=q$. A state $q\in Q$ is called \emph{coaccessible} if there exist a string $s\in\Sigma^*$ and an accepting state $q_f\in F$ such that $\hat{\Delta}(q,s) = q_f$. Equivalently, $q\in Q$ is coaccessible if $\mathcal{L}(q)\cap F\neq\emptyset$. If all states of an automaton are accessible and coaccessible, it is called \emph{trim}.

Let us briefly review the classical textbook construction of a DFA recognizing the same language as a given NFA.
\begin{lem}[Subset Construction; Rabin and Scott,~\cite{RabSco59}]
Let $M=(Q, \Sigma, \Delta, Q_\alpha, F)$ be an NFA. Then $(\powset{Q}, \Sigma, \delta, Q_\alpha, \lbrace Q'\in\powset{Q}\vert Q'\cap F\neq\emptyset\rbrace)$, with $\delta:(Q',\sigma)\mapsto\bigcup_{q'\in Q'}\Delta(q',\sigma)$, is a DFA that recognizes the same language as $M$.
\end{lem}
\begin{proof}
Omitted. See~\cite{RabSco59} or~\cite{Kozen1999Automata}.
\end{proof}

As mentioned above, some DFA states may be inaccessible. These states can be removed from the DFA's state space. To ease notation, we write $\proc{SubsetConstruction}(M)$ to denote the DFA resulting from the subset construction and subsequent removal of inaccessible states.

\section{Simple NFAs}\label{sec:simple_nfa}
Recall that our goal is to identify a class of NFAs for which the subset construction yields a minimal DFA; where a DFA is called minimal if there does not exist a DFA with fewer states that recognizes the same language. To this end, we define \emph{simple NFAs}.

\begin{defn}[Simple non-deterministic finite automaton]\label{def:simple_nfa}
Let an NFA~$M=(Q, \Sigma, \Delta, q_\alpha, F)$ be given. $M$ is called \emph{simple} if all states are accessible and the languages~$\mathcal{L}(q)$ of all states $q\in Q$ are non-empty and pairwise disjoint.
\end{defn}

Therefore, an automaton is simple if and only if it is trim and the languages of all states are pairwise disjoint. Note that an automaton can easily be made trim: If there is a state $q$ that is not coaccessible, that is, $\mathcal{L}(q)$ is empty, we can safely remove $q$ from $Q$ without changing the recognized language. Likewise, all inaccessible states can be removed without changing the recognized language.

\begin{thm}[Minimality of DFA constructed from simple NFA]\label{thm:simple_nfa}
Let $M_n=(Q, \Sigma, \Delta, Q_\alpha, F)$ be a simple NFA. Then, the DFA
\[M_d=\big(\mathcal{Q}\subset\powset{Q}, \Sigma, \delta, Q_\alpha, \mathcal{F}\big)=\proc{SubsetConstruction}(M_n)\]
is minimal.
\end{thm}

Before we are able to prove this, we need an auxiliary lemma and the notion of \emph{equivalent states} in a DFA. We define two states $p$ and $q$ of a DFA $(Q', \Sigma', \delta', q_\alpha', F')$ to be \emph{equivalent} if
$\hat{\delta}'(p,s)\in F'\Longleftrightarrow\hat{\delta}'(q,s)\in F'$ for all~$s\in\Sigma^\ast$.
\begin{lem}\label{lem:minimality}
A DFA is minimal if and only if its states are pairwise non-equivalent.
\end{lem}
\begin{proof}
See Chapters~13 and~15 in~\cite{Kozen1999Automata}.
\end{proof}

\begin{proof}[Proof of Theorem~\ref{thm:simple_nfa}]
Let $Q', Q''\in\mathcal{Q}$ be two distinct DFA states. By Lemma~\ref{lem:minimality}, we have to show that $Q'$ and $Q''$ are not equivalent, or more formally
\begin{equation}\label{eqn:states_equivalent}
\mathcal{L}(Q')=\bigcup_{q'\in Q'}\mathcal{L}(q')\neq\bigcup_{q''\in Q''}\mathcal{L}(q'')=\mathcal{L}(Q'')\mbox{\,.}
\end{equation}
Without loss of generality, assume that $Q'\setminus Q''\neq\emptyset$ and let $q\in Q'\setminus Q''$. By Definition~\ref{def:simple_nfa}, $\mathcal{L}(q)\cap \mathcal{L}(q'')=\emptyset$ for all $q''\in Q''$ and thus $\mathcal{L}(q)\cap\mathcal{L}(Q'')=\emptyset$. But, by choice of $q$, $\mathcal{L}(q)\subset\mathcal{L}(Q')$ and, by Definition~\ref{def:simple_nfa}, $\mathcal{L}(q)\neq\emptyset$. Hence, it follows that $\mathcal{L}(Q')\neq\mathcal{L}(Q'')$.
\end{proof}

\subsection{An Alternative Proof}\label{sec:alt_proof}
We give an alternative proof of Theorem~\ref{thm:simple_nfa} by means of the theory developed in~\cite{GlaPlo08}. There, van~Glabbeek and Ploeger consider five different variants of the classical subset construction. Each variant is characterized by an operation $f:\powset{Q}\rightarrow\powset{Q}$, where $Q$ is the state space of an NFA. When a new DFA state is produced in the course of the subset construction, it is subjected to the operation~$f$ before being added to the final automaton. In one variant, they define~$f$ to be the closure operation
\[\closure:Q'\mapsto\big\lbrace q\in Q\,\big\vert\,\mathcal{L}(q)\subseteq\mathcal{L}(Q')\big\rbrace\]
and show that the subset construction endowed with this operation directly produces minimal DFAs. Theorem~\ref{thm:simple_nfa} now follows from the definition of simple NFAs: As all sets $\mathcal{L}(q)$ for $q \in Q$ are pairwise disjoint, $\closure(Q')=Q'$ for each $Q'\subseteq Q$ and, thus, the classical subset construction yields a minimal DFA.

Note that the language inclusion problem required to be solved for the $\closure$-operation is in general hard to compute. According to~\cite{GlaPlo08}, it is PSPACE-complete.

\subsection{Self-Transitions of Start States}\label{sec:self-transition}
In most practical settings like pattern search or pattern statistics, we are given a certain type of pattern and need to construct an automaton that accepts all strings with a suffix matching this pattern, rather than an automaton that accepts only the strings that match the pattern. For instance, if our pattern is the single string \texttt{ABC} and we want to find all occurrences of \texttt{ABC} in a long text, we need to build an automaton recognizing all strings whose last three letters are \texttt{ABC}. For NFAs, we can easily obtain such an automaton once we have constructed an NFA accepting all strings that match our pattern. All we need to do is to modify the transition function~$\Delta$ by adding self-transitions to all start states
\begin{equation}\label{eqn:self-transition}
\Delta_\circlearrowleft:(q,\sigma)\mapsto
\begin{cases}
 \lbrace q\rbrace\cup \Delta(q,\sigma) & \mbox{if }q\in Q_\alpha\mbox{\,,} \\
 \Delta(q,\sigma) & \mbox{otherwise\,.}
\end{cases}
\end{equation}
Throughout this article, the subscript~``$\circlearrowleft$'' refers to this modification of a transition function. The next Lemma characterizes those simple NFAs that remain simple under this modification.

\begin{lem}\label{lem:self-transition}
Let $M=(Q, \Sigma, \Delta, Q_\alpha, F)$ be a simple NFA. The modified automaton $M_\circlearrowleft:=(Q, \Sigma, \Delta_\circlearrowleft, Q_\alpha, F)$ is simple if and only if, in~$M$, no start state can be reached from any other state. That means there do not exist $\sigma\in\Sigma$, $q_\alpha\in Q_\alpha$, and $q\in Q$ with $q_\alpha\neq q$ such that $q_\alpha\in\Delta(q,\sigma)$.
\end{lem}

\begin{proof}
In this proof, we use the notation $\mathcal{L}_\circlearrowleft(q)$ to refer to the language of the state~$q$ with respect to the modified NFA $(Q, \Sigma, \Delta_\circlearrowleft, Q_\alpha, F)$.

``$\Longrightarrow$'': Suppose $(Q, \Sigma, \Delta_\circlearrowleft, Q_\alpha, F)$ is simple and there exist $\sigma\in\Sigma$, $q_\alpha\in Q_\alpha$, and $q\in Q$ with $q_\alpha\neq q$ such that $q_\alpha\in\Delta(q,\sigma)$. Thus, $\sigma s\in\mathcal{L}(q)$ for all $s\in\mathcal{L}(q_\alpha)$. Because of the added self-transition, we also have $\sigma s\in\mathcal{L}_\circlearrowleft(q_\alpha)$ and, thus, $\mathcal{L}_\circlearrowleft(q_\alpha)$ and $\mathcal{L}_\circlearrowleft(q)$ are not disjoint, contradicting the assumption that $M_\circlearrowleft$ is simple.

``$\Longleftarrow$'': Now, we assume that there do not exist any~$\sigma\in\Sigma$, $q_\alpha\in Q_\alpha$, and $q\in Q$ with $q_\alpha\neq q$ such that $q_\alpha\in\Delta(q,\sigma)$. The properties that all states are accessible and coaccessible cannot get lost by adding the additional self-transitions. Therefore, we only need to verify that $\L_\circlearrowleft(q)$ and $\L_\circlearrowleft(q')$ are disjoint for all distinct $q,q'\in Q$. For the sake of contradiction, we assume there exist distinct $q,q'\in Q$ violating this condition. We choose $s\in\L_\circlearrowleft(q)\cap\L_\circlearrowleft(q')$ such that $s\in\L_\circlearrowleft(q)\setminus\L(q)$; if that is not possible, it becomes possible after swapping $q$ and $q'$, because $\L(p)\subseteq\L_\circlearrowleft(p)$ for all $p\in Q$ and $\L(q)\cap\L(q')=\emptyset$. We have to distinguish two cases:

\emph{Case 1} ($s\in\L(q')$): By our assumption, there does not exist a state in $Q\setminus Q_\alpha$ from which a start state can be reached. This means that the transition function remains unchanged for all states reachable from any state in $Q\setminus Q_\alpha$, which implies that $\L(p)=\L_\circlearrowleft(p)$ for all $p\in Q\setminus Q_\alpha$. Therefore, $q$ must be a start state. We chose $s$ to lie in $\L_\circlearrowleft(q)\setminus\L(q)$, which implies that there exists a $k\in\N$ such that $\suffix{s}{k}\in\L(q)$. 
Since all $\L(p)$ for $p\in Q$ are disjoint, it follows that $\suffix{s}{k}\notin\L(p)$ for all $p\in Q\setminus\lbrace q\rbrace$. As $s\in\L(q')$, we thus conclude that $\Delta(q',\prefix{s}{k-1})=q$, which contradicts the assumption that we cannot reach a start state from any other state than itself.

\emph{Case 2} ($s\notin\L(q')$): By the same argument as in the last case, we conclude that $q$ and $q'$ must be start states. Again, this implies the existence of $k,k'\in\N$ such that $\suffix{s}{k}\in\L(q)$ and $\suffix{s}{k'}\in\L(q')$. If $k=k'$, then $\suffix{s}{k}\in\L(q)\cap\L(q')\neq\emptyset$, contradicting the simpleness of $M$. We assume, without loss of generality, that $k<k'$. Since $\suffix{s}{k'}\in\L(q')$ and $\suffix{s}{k'}\notin\L(p)$ for all $p\in Q\setminus\lbrace q'\rbrace$, we conclude that $\Delta(q,\substr{s}{k}{k'-1})=q'$, again contradicting the assumption that we cannot reach a start state from any other state than itself.
\end{proof}

\section{Application to Generalized Strings}\label{sec:app_to_gen_strings}

In the next two sections, we show that generalized strings and sets of generalized strings admit the construction of simple NFAs. Obviously, a single string is a special case of a set of strings. To aid understandability, we nonetheless start with the easier case of one single string.

\subsection{Single Generalized Strings}\label{sec:single_gen_string}
For a generalized string $g$, an NFA recognizing all strings that match $g$ can easily be constructed by connecting the state set $Q=\lbrace 0,\ldots,|g|\rbrace$ with the transition function
\[\Delta:(q,\sigma)\mapsto
\begin{cases}
 \lbrace q+1\rbrace & \mbox{if }q<\len{g}\mbox{ and } \sigma\in \chr{g}{q+1}\mbox{\,,} \\
 \emptyset & \mbox{otherwise\,.}
\end{cases}
\]

\begin{figure}
\begin{center}
\includegraphics[width=\textwidth]{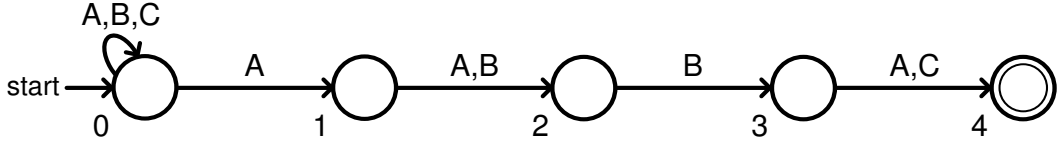}
\end{center}
\caption{Example of a simple NFA (with self-transition added to the start state) constructed from the generalized string \texttt{\{A\}\{A,B\}\{B\}\{A,C\}} over the alphabet $\Sigma=\texttt{\{A,B,C\}}$. The accepting state is represented by two concentric circles.
}
\label{fig:nfa_one_genstring}
\end{figure}

Setting $Q_\alpha=\lbrace 0\rbrace$ and $F=\lbrace \len{g} \rbrace$ completes the construction of our NFA $(Q, \Sigma, \Delta, Q_\alpha, F)$. For brevity, we write $\proc{NFA}(g)$ to denote the automaton created from a generalized string~$g$ using the above construction.

\begin{lem}\label{lem:single_gen_string_nfa_simple}
Let~$g$ be a generalized string. Then $M_g:=\proc{NFA}(g)$ is a simple NFA.
\end{lem}
\begin{proof}
Clearly, all states $i\in Q$ are accessible and coaccessible. $M_g$ admits only transitions from a state $i$ to its successor state $i+1$; only the last state in this chain is an accepting state. Thus, for each state $i\in Q$, the lengths of all accepted strings $s\in\mathcal{L}(i)$ equal $\len{g}-i$. Hence, for two different states~$i$ and~$j$, accepted strings have different lengths. Thus, all $\mathcal{L}(i)$ must be pairwise disjoint (for $i\in Q$).
\end{proof}

As discussed in Section~\ref{sec:self-transition}, we often need to add a self-transition to the start state. This modification is defined formally in Equation~\eqref{eqn:self-transition}. We write $\proc{NFA}_\circlearrowleft(g)$ to refer to the resulting automaton. See Figure~\ref{fig:nfa_one_genstring} for an example. Combining Theorem~\ref{thm:simple_nfa}, Lemma~\ref{lem:single_gen_string_nfa_simple}, and Lemma~\ref{lem:self-transition}, we arrive at the following corollary:

\begin{cor}
Let $g$ be a generalized string and $M_g:=\proc{NFA}_\circlearrowleft(g)$ the corresponding NFA. Then, $\proc{SubsetConstruction}(M_g)$
is a minimal DFA.
\end{cor}

\subsection{Sets of Generalized Strings}\label{sec:gen_string_set}
In this section, we generalize the above results to finite sets of generalized strings of equal length. Speaking formally, we assume a length~$\ell$ and $G\subset\mathcal{G}^\ell$ to be given and seek to construct a simple NFA that recognizes all strings that have a suffix matching a $g\in G$. As above, we first construct an automaton that recognizes all strings matching a $g\in G$ and, in a second step, add self-transitions to the start states~$Q_\alpha$.

The automaton we build is organized level-wise with $\ell+1$ levels. Transitions are only possible between states in adjacent levels and only in one direction (which we choose to call \emph{downwards}). The bottom level contains just one state which is the single accepting state; all states in the top level are start states. As before for a single generalized string, two states $q'$ and $q''$ in different levels are obviously ``language-disjoint'', meaning that $\mathcal{L}(q')\cap\mathcal{L}(q'')=\emptyset$. But here, we possibly need more than one state in a level, which entails the problem of ensuring language-disjointness for states in the same level. We achieve this by using a state space induced by a special parent-child relation between states in adjacent levels. Before we formally construct state space and automaton, the impatient reader may have a look at the example in Figure~\ref{fig:nfa_three_genstrings}.

\begin{figure}[t]
\begin{center}
\includegraphics[width=\textwidth]{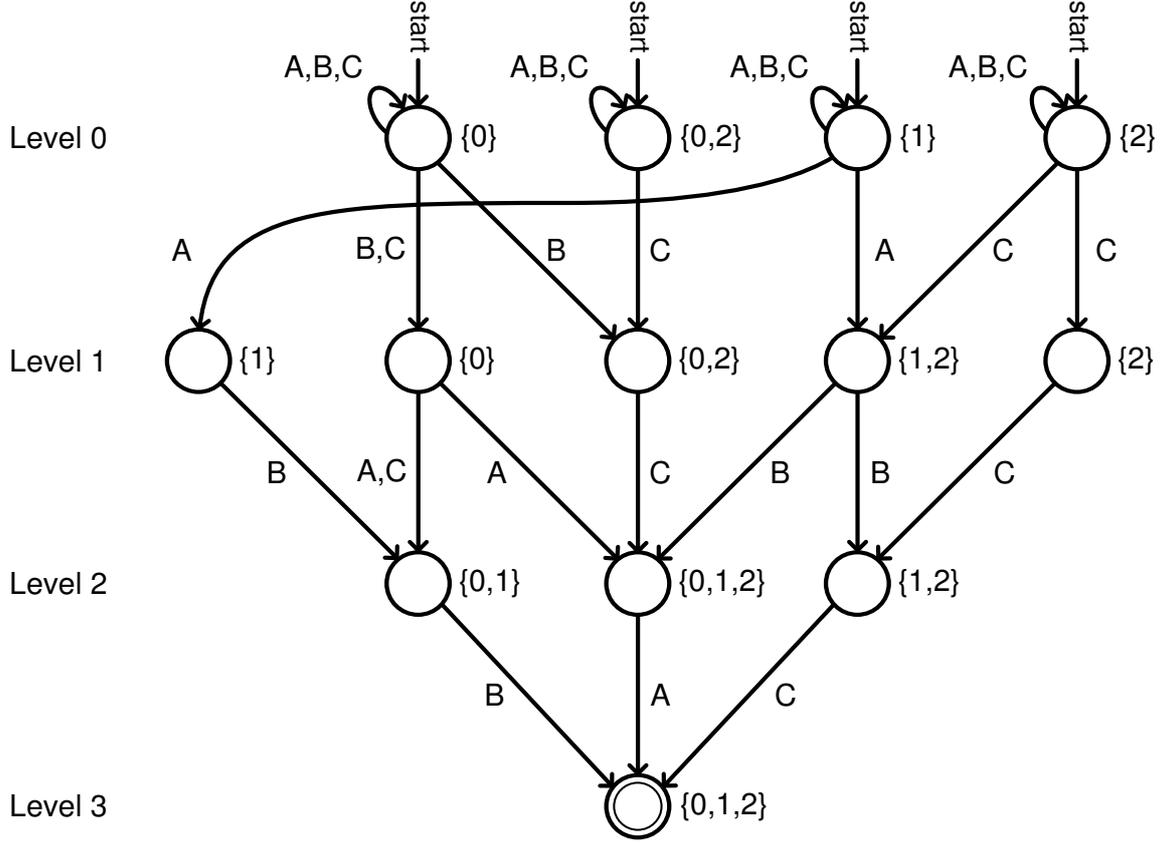}
\end{center}
\caption{Example of a simple NFA constructed from the three generalized strings 0:\texttt{\{B,C\}\{A,C\}\{A,B\}}, 1:\texttt{\{A\}\{B\}\{A,B,C\}}, and 2:\texttt{\{C\}\{B,C\}\{A,C\}} over the alphabet $\Sigma=\texttt{\{A,B,C\}}$. Each state is annotated with the set of generalized strings that are ``active'' in this state (each generalized string is represented by its index 0, 1, or 2). The accepting state is represented by two concentric circles.
}
\label{fig:nfa_three_genstrings}
\end{figure}

Let us begin with the formal specification of a suitable state space $Q$. We choose $Q$ to be a special subset of $\bar{Q}:=\powset{G}\times\lbrace 0,\ldots,\ell\rbrace$ with the following semantics in mind: to be in state $q=(H,k)$ means that the last~$k$ characters read match the first~$k$ positions of a $g\in H$. For the definition of~$Q$, we need the function $\proc{Parent}:\bar{Q}\times\Sigma\rightarrow\bar{Q}\cup\lbrace\bot\rbrace$ given by
\begin{equation}\label{eqn:gen_string_parent}
\proc{Parent}:\big((H,k),\sigma\big)\mapsto
\begin{cases}
\big(\,\lbrace h\in H\,\vert\,\sigma\in\chr{h}{k}\rbrace\mbox\,,\,k-1\,\big) & \mbox{if }k>0\mbox{\,,} \\
\bot & \mbox{otherwise\,.}
\end{cases}
\end{equation}
We say that $\proc{Parent}(q,\sigma)$ is a parent of~$q$ under the character~$\sigma$. The special symbol~$\bot$ is used to indicate that a state is in the top level and therefore does not have any parents. The \proc{Parent} mapping induces a hierarchy of $\ell+1$ levels of states:
\begin{align}
\label{eqn:gen_string_Qell}Q_\ell:=&\lbrace(G,\ell)\rbrace\mbox{\,,}\\
\label{eqn:gen_string_Qi}Q_i:=&\Big\lbrace(H,i), H\in\powset{G}\setminus\lbrace\emptyset\rbrace\,\Big\vert\,\exists q\in Q_{i+1},\sigma\in\Sigma:\proc{Parent}(q,\sigma)=(H,i)\Big\rbrace\mbox{\,,}
\end{align}
for $0\leq i<\ell$. Finally, we write our state space as
\begin{equation}\label{eqn:gen_strings_state_space}
Q:=Q_0\cup\ldots\cup Q_\ell\mbox{\,.}
\end{equation}
The \proc{Parent} mapping also induces a transition function~$\Delta$:
\begin{equation}\label{eqn:gen_string_delta}
\Delta:\big((H,k),\sigma\big)\mapsto
\begin{cases}
\big\lbrace q\in Q_{k+1}\,\big\vert\,\proc{Parent}(q,\sigma)=(H,k)\big\rbrace & \mbox{if }k<\ell\mbox{\,,}\\
\emptyset & \mbox{otherwise\,.}
\end{cases}
\end{equation}
To complete the construction, we set $Q_\alpha:=Q_0$ and $F:=Q_\ell=\lbrace(G,\ell)\rbrace$ and obtain $\proc{NFA}(G):=(Q, \Sigma, \Delta, Q_\alpha, F)$.
The next lemma states that an NFA constructed in this way accepts exactly the language given by~$G$.
\begin{lem}
Let a length $\ell\in\N$, a set of generalized strings $G\subset\mathcal{G}^\ell$, and $(Q, \Sigma, \Delta, Q_\alpha, F)=\proc{NFA}(G)$ be given. Then,
\[\exists q\in Q_\alpha:\hat{\Delta}(q,s)\cap F\neq\emptyset\mboxspace{$\Longleftrightarrow$}\exists g\in G: s\matches g\mbox{\,,}\]
for all $s\in\Sigma^\ast$.
\end{lem}
\begin{proof}
We start with the forward direction ``$\Longrightarrow$''. If $s\in\Sigma^\ast$ is accepted by $\proc{NFA}(G)$, then there exists a sequence of states $q_0,\ldots,q_{\len{s}}$ such that $q_0\in Q_\alpha$, $q_{\len{s}}\in F$, and $q_i\in\Delta(q_{i-1},\chr{s}{i})$ for $0<i\leq\len{s}$. It follows from Equation~\eqref{eqn:gen_string_delta} that $q_{i-1}=\proc{Parent}(q_i,\chr{s}{i})$. Hence, Equation~\eqref{eqn:gen_string_parent} implies that $H_0\subset\ldots\subset H_{\len{s}}$, where $(H_i,k_i):=q_i$. Furthermore, by Equation~\eqref{eqn:gen_string_Qi}, $H_0$ is non-empty. Inductively applying~\eqref{eqn:gen_string_parent} now yields that~$s\matches h$ for all $h\in H_0$, which proves the forward direction.

Let us prove the backward direction ``$\Longleftarrow$''. Let $g\in G$, such that $s\matches g$. Consider the sequence of states $q'_0,\ldots ,q'_\len{s}$ with $(H'_i,k'_i):=q'_i$ given by $q'_{\len{s}}:=(G,\ell)$ and $q'_{i-1}:=\proc{Parent}(q'_i,\chr{s}{i})$ for $0<i\leq\len{s}$. From $s\matches g$ and Equation~\eqref{eqn:gen_string_parent} it follows that $g\in H'_i$ for $0\leq i\leq\len{s}$. Thus, each $H'_i$ is non-empty and by Equations~\eqref{eqn:gen_string_Qell} and~\eqref{eqn:gen_string_Qi} we get $q'_i\in Q_i$ for $0\leq i\leq\len{s}$, implying that $q'_0\in Q_0=Q_\alpha$ is a start state. From Equation~\eqref{eqn:gen_string_delta} we conclude that $\hat{\Delta}(q'_0,s)=q'_{\len{s}}$ which proves the claim as $q'_{\len{s}}\in Q_\ell=F$.
\end{proof}

In analogy to Lemma~\ref{lem:single_gen_string_nfa_simple}, we verify that $\proc{NFA}(G)$ is indeed a simple NFA.

\begin{lem}\label{lem:gen_string_set_nfa_simple}
Let $\ell\in\N$ and $G\subset\mathcal{G}^\ell$. Then, $M_G:=\proc{NFA}(G)$ is a simple NFA.
\end{lem}
\begin{proof}
The level-wise construction directly implies that all states are accessible and coaccessible, i.e.\ $\mathcal{L}(q)$ is non-empty for all~$q\in Q$. States with empty~$\mathcal{L}(q)$ cannot be generated by Equation~\eqref{eqn:gen_string_Qi}.

It remains to be shown that for all distinct~$p,q\in Q$ the sets~$\mathcal{L}(p)$ and~$\mathcal{L}(q)$ are disjoint. By construction, this is clearly true if~$p$ and~$q$ are in different levels. Hence, it suffices to show that
\begin{equation}\label{eqn:levelwise_disjointness}
\mathcal{L}(p)\cap\mathcal{L}(q)=\emptyset\mbox{ for all }p,q\in Q_i\mbox{ with }p\neq q
\end{equation}
for all $Q_i$ with $0\leq i\leq\ell$.
We prove this by induction on $i$. First, note that for $i=\ell$, Condition~\eqref{eqn:levelwise_disjointness} is fulfilled as $|Q_\ell|=1$. Assume that~\eqref{eqn:levelwise_disjointness} holds for $i>0$. For the sake of contradiction, we further assume there exist distinct $p,q\in Q_{i-1}$, such that $\mathcal{L}(p)\cap\mathcal{L}(q)\neq\emptyset$. Let $s\in\mathcal{L}(p)\cap\mathcal{L}(q)$; it follows that $\hat{\Delta}(p,s)\in F$. There must exist a state~$r\in Q_i$ such that $\hat{\Delta}(r,\suffix{s}{2})\in F$. As, by our induction hypothesis, Condition~\eqref{eqn:levelwise_disjointness} holds for~$i$, we conclude that the state~$r$ is unique. It follows from~\eqref{eqn:gen_string_delta} that $r\in\Delta(p,\chr{s}{1})$ and $r\in\Delta(q,\chr{s}{1})$. Applying the definition of~$\Delta$, we get $p=\proc{Parent}(r,\chr{s}{1})=q$ and, thus, $p=q$.
\end{proof}

In Section~\ref{sec:single_gen_string}, we added an initial self-transition to the constructed NFA in order to accept not only the given generalized string, but all strings whose suffix matches the generalized string. We thereby obtained an automaton that finds all occurrences of the generalized string in a given text. Now we repeat this step by transforming $\proc{NFA}(G)$ using Equation~\eqref{eqn:self-transition}. Again, we refer to the resulting modified automaton by $\proc{NFA}_\circlearrowleft(G)$. Note that for $|G|=1$ we obtain the same automaton as constructed in Section~\ref{sec:single_gen_string}. Combining Theorem~\ref{thm:simple_nfa}, Lemma~\ref{lem:gen_string_set_nfa_simple}, and Lemma~\ref{lem:self-transition} yields the following corollary:

\begin{cor}
Let $\ell\in\N$, $G\subset\mathcal{G}^\ell$, and $M_G:=\proc{NFA}_\circlearrowleft(G)$.

Then, the result of $\proc{SubsetConstruction}(M_G)$ is a minimal DFA.
\end{cor}

\subsubsection{Algorithm and Runtime}
The construction scheme formalized in Equations~\eqref{eqn:gen_string_Qell} and~\eqref{eqn:gen_string_Qi} can directly be translated into an algorithm:
\begin{enumerate}
\item Initialize transition map $\Delta$ to be empty.
\item Initialize the bottom level $Q_\ell$ to contain its only state $(G,\ell)$.
\item\label{step:loop_level} For $k$ from $\ell-1$ down to $0$, build level~$Q_k$:
  \begin{enumerate}
    \item Initialize level~$Q_k$ to be empty.
    \item\label{step:loop_node_alphabet} For each node $(H',k+1)\in Q_{k+1}$ and each $\sigma\in\Sigma$
       \begin{enumerate}
         \item\label{step:compute_H} Compute the set $H:=\big\lbrace h\in H'\,\big\vert\,\sigma\in\chr{h}{k+1}\big\rbrace$.
         \item If $H\neq\emptyset$ and $(H,k)\notin Q_k$, add $(H,k)$ to $Q_k$.
         \item Add transition $\big((H,k),\sigma\big)\mapsto (H',k+1)$ to $\Delta$.
       \end{enumerate}
  \end{enumerate}
\item Add self-transitions to all $q\in Q_0$.
\end{enumerate}

In Loop~\ref{step:loop_level}, we build $\ell$ levels. Each level contains at most $2^{|G|}$ states and thus the body of Loop~\ref{step:loop_node_alphabet} is executed $\OO(2^{|G|}\cdot|\Sigma|)$ times for each level, where Step~\ref{step:compute_H} takes~$\OO(|G|)$ time and the other steps can be performed in constant time. All in all, the algorithm takes~$\OO(2^{|G|}\cdot \ell\cdot|\Sigma|\cdot|G|)$ time.

The construction of a minimal DFA from a set of generalized strings thus takes~$\OO(2^{|G|}\cdot \ell\cdot|\Sigma|\cdot|G|+m)$ time, where~$m$ is the number of states in the minimal DFA. 

\section{Application to Consensus Strings with a Hamming Neighborhood}\label{sec:app_to_consensus_strings}
Another type of motif commonly used in computational biology is a consensus string along with a distance threshold. Here, we assume that a (generalized) string~$s$ and a distance threshold~$d_{\max}$ are given and want to compute the minimal DFA that recognizes all strings with a Hamming distance to~$s$ of at most~$d_{\max}$, where the Hamming distance between a string~$s$ and a generalized string~$g$ of same length is defined as
\[
d(s,g):=\Big\vert\Big\lbrace i\in\lbrace 1,\ldots,\len{s}\rbrace\,\big\vert\,\chr{s}{i}\notin\chr{g}{i}\Big\rbrace\Big\vert\mbox{\,.}
\]
\begin{figure}[t!]
\begin{center}
\includegraphics[width=\textwidth]{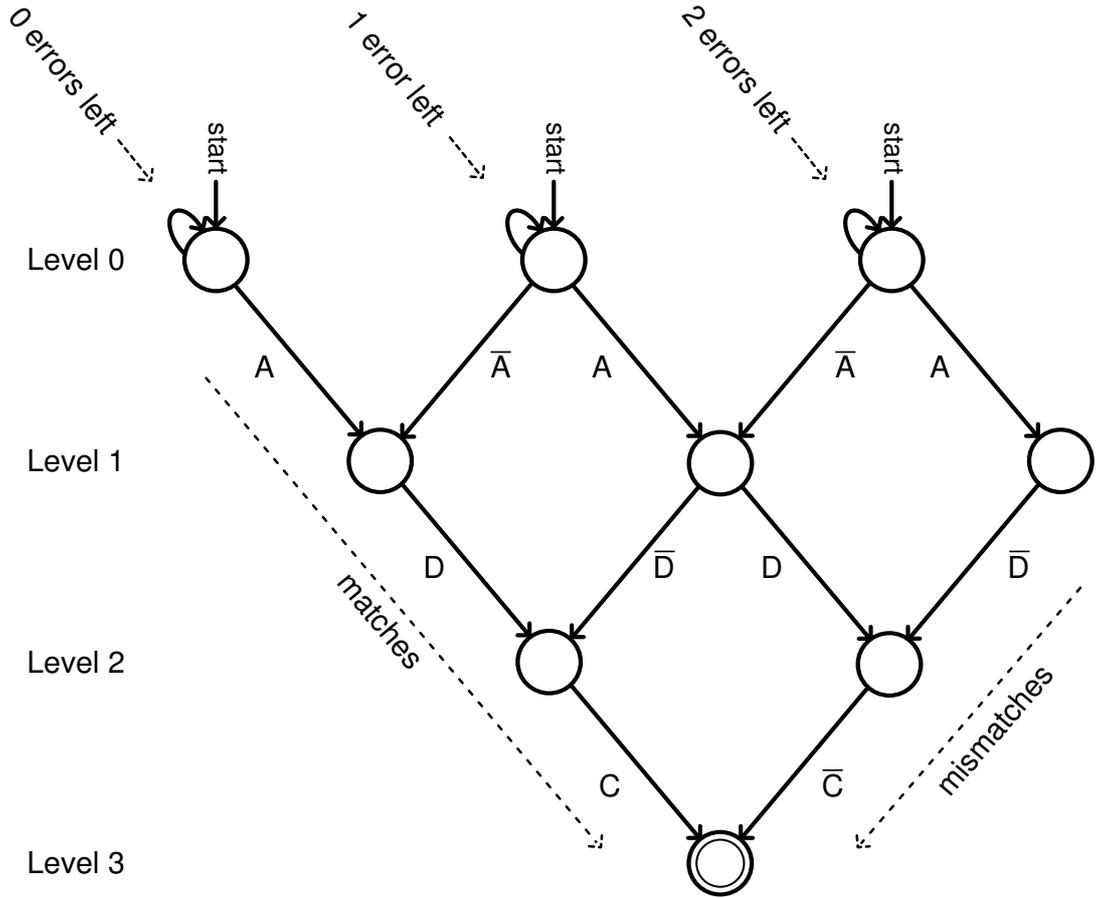}
\end{center}
\caption{Example of a simple NFA over the alphabet $\Sigma=\lbrace\texttt{A,B,C,D}\rbrace$ recognizing the consensus \texttt{ADC} and all strings within a Hamming distance of two or less. Characters with bars stand for the inverse, e.g. $\overline{\texttt{A}}$ stands for \texttt{B}, \texttt{C}, or \texttt{D}. The accepting state is represented by two concentric circles.
}
\label{fig:nfa_consensus}
\end{figure}
In this section, we construct a simple NFA recognizing a generalized string and its Hamming neighborhood. The construction is similar to the one given in~\cite{Navarro2002Flexible}. Interestingly, the resulting NFA turns out to be simple.

The basic idea for the construction is to use a two-dimensional grid as a state space, where we advance into one dimension whenever a valid character has been read and into the other dimension for each mismatch. Figure~\ref{fig:nfa_consensus} illustrates an NFA built in this way. Formally the state space is defined by
\begin{equation}\label{eqn:nfa_consensus_Q}
Q:=\Big\lbrace(e,k)\in\lbrace 0,\ldots,d_{\max}\rbrace\times\lbrace 0,\ldots,\len{g}\rbrace\,\Big\vert\,\len{g}-k-e\geq 0\Big\rbrace
\end{equation}
with the following semantics: state $(e,k)$ accepts all strings of length $\len{g}-k$ that match the respective suffix of~$g$ with exactly~$e$ errors. The condition $\len{g}-k-e\geq 0$ states that the number of errors~$e$ cannot be larger than $\len{g}-k$, which is the number of characters left. We define the transition function to obey this semantics:
\begin{equation}\label{eqn:nfa_consensus_Delta}
\Delta:(e,k)\times\sigma\mapsto
\begin{cases}
z(e,k+1) & \mbox{if $\sigma\in\chr{g}{k+1}$\,,}\\
z(e-1,k+1) & \mbox{otherwise\,,}
\end{cases}
\end{equation}
where the function $z:\Z\times\Z\rightarrow\powset{Q}$ returns the empty set whenever we ``fall off the grid''. More precisely,
\begin{equation}\label{eqn:nfa_consensus_z}
z:(e,k)\mapsto
\begin{cases}
\big\lbrace(e,k)\big\rbrace & \mbox{if $(e,k)\in Q$\,,} \\
\emptyset & \mbox{otherwise\,.}
\end{cases}
\end{equation}
As before, the topmost level constitutes the start states, i.e.\ $Q_\alpha:=\big\lbrace(e,k)\in Q\,\vert\,k=0\big\rbrace$, and the bottommost level contains only the single accepting state, i.e.\ $F:=\lbrace(0,\len{g})\rbrace$. We write $\proc{NFA}(g,d_{\max}):=(Q, \Sigma, \Delta, Q_\alpha, F)$ to denote the NFA constructed in this way. Again, we use the notation $\proc{NFA}_\circlearrowleft(g,d_{\max}):=(Q, \Sigma, \Delta_\circlearrowleft, Q_\alpha, F)$ to refer to the automaton with self-transitions added to the start states. Note that for $d_{\max}=0$, the resulting automaton is isomorphic to the one constructed from a single generalized string in Section~\ref{sec:single_gen_string}.

In order to prove that the construction is correct and produces simple NFAs, we use the following Lemma on the state's languages.

\begin{lem}\label{lem:nfa_consensus_state_property}
Let $g\in \mathcal{G}_\Sigma^\ast$, $d_{\max}\in\N_0$ and $M=\proc{NFA}(g,d_{\max})=(Q, \Sigma, \Delta, Q_\alpha, F)$. Then, the language of state $(e,k)$ is characterized by
\[
\mathcal{L}\big((e,k)\big) = \Big\lbrace s\in\Sigma^{\len{g}-k}\,\Big\vert\,d\big(s,\suffix{g}{k+1}\big)=e\Big\rbrace\mbox{\,,}
\]
for all $(e,k)\in Q$.
\end{lem}
\begin{proof}
We start with the direction ``$\subseteq$''. By construction of~$\Delta$ and~$F$, we have $\mathcal{L}\big((e,k)\big)\subseteq\Sigma^{\len{g}-k}$. Let $s\in\mathcal{L}\big((e,k)\big)$, then $\hat{\Delta}\big((e,k),s\big)=(0,\len{g})$. That means, in the course of $\len{s}$ state transitions the first component of the state changes from~$e$ to~0. As we see from Equation~\eqref{eqn:nfa_consensus_Delta}, the only change possible in the first component is a decrease by~1, which happens if and only if the read character is a mismatch. Thus, it follows that $d\big(s,\suffix{g}{k+1}\big)=e$.

Now we prove the backward direction ``$\supseteq$''. Let $s\in\Sigma^{\len{g}-k}$ and $d\big(s,\suffix{g}{k+1}\big)=e$. That means there are exactly $e$ indices $a_1,\ldots,a_e$ such that $\chr{s}{a_i}\notin\chr{g}{k+a_i}$ for $1\leq i\leq e$. Provided that all states exist and thus the~$z$ function never returns~$\emptyset$, we apply the first case of~\eqref{eqn:nfa_consensus_Delta} exactly~$\len{s}-e$ times and the second case exactly $e$ times, ending in state $(0,\len{g})$ as claimed. The only thing left to verify is that~$z$ indeed never returns~$\emptyset$. Note that, by~\eqref{eqn:nfa_consensus_Delta}, the term $\len{g}-k-e$ cannot increase. Since it reaches zero after~$\len{s}$ steps, it cannot have been smaller than zero at any time. Hence, by Equation~\eqref{eqn:nfa_consensus_Q}, all intermediate states exist and, thus, the first case of Equation~\eqref{eqn:nfa_consensus_z} is applied for all state transitions.
\end{proof}

Using this lemma, the construction's correctness is easily verified:
\begin{lem}
Let $g\in \mathcal{G}_\Sigma^\ast$, $d_{\max}\in\N_0$ and $M=\proc{NFA}(g,d_{\max})=(Q, \Sigma, \Delta, Q_\alpha, F)$. Then, $M$ accepts exactly the strings $\lbrace s\in\Sigma^\len{g}\,\vert\,d(s,g)\leq d_{\max}\rbrace$.
\end{lem}
\begin{proof}
By definition, $M$ accepts the strings $\L(Q_\alpha)$. By construction of~$Q_\alpha$ and Lemma~\ref{lem:nfa_consensus_state_property}, we obtain
\[
\L(Q_\alpha)=\bigcup_{e=0}^{\min(d_{\max},\len{g})} \L\big((e,0)\big) = \bigcup_{e=0}^{\min(d_{\max},\len{g})}\big\lbrace s\in\Sigma^\len{g}\,\big\vert\,d(s,g)=e\big\rbrace
\]
\end{proof}

\begin{lem}
Let $g\in \mathcal{G}_\Sigma^\ast$, $d_{\max}\in\N_0$. Then, $\proc{NFA}(g,d_{\max})=(Q, \Sigma, \Delta, Q_\alpha, F)$, is a simple NFA.
\end{lem}
\begin{proof}
By construction, all states are accessible and coaccessible. The disjointness of $\L(q)$ and $\L(q')$ for distinct~$q,q'\in Q$ follows immediately from Lemma~\ref{lem:nfa_consensus_state_property}.
\end{proof}

In analogy to Sections~\ref{sec:single_gen_string} and~\ref{sec:gen_string_set}, we can now add self-transitions to the start states to obtain $\proc{NFA}_\circlearrowleft(g,d_{\max})$. Note that, again, the conditions of Lemma~\ref{lem:self-transition} are satisfied, allowing us to apply Theorem~\ref{thm:simple_nfa}.

\begin{cor}
Let $g\in \mathcal{G}_\Sigma^\ast$, $d_{\max}\in\N_0$, and $M=\proc{NFA}_\circlearrowleft(g,d_{\max})$. Then, the result of $\proc{SubsetConstruction}(M)$ is a minimal DFA.
\end{cor}

The state space of $\proc{NFA}_\circlearrowleft(g,d_{\max})$ has a size of~$\OO(\len{g}\cdot d_{\max})$. Deriving a construction algorithm that uses~$\OO(1)$ time per state is straightforward. We can, therefore, construct the minimal DFA from a generalized string~$g$ and the distance threshold~$d_{\max}$ in time~$\OO(\len{g}\cdot d_{\max} + m)$, where~$m$ is the size of the minimal DFA.

\section{Conclusions}
We introduced the concept of simple NFAs. These automata have a useful property: when subjected to the standard subset construction, they result in minimal DFAs. Motivated by a background in bioinformatics, we turned our attention to pattern classes found in this field. We gave an algorithm to construct a simple NFA from a set~$G$ of generalized strings of equal length~$\ell$ in time~$\OO(2^{|G|}\cdot \ell\cdot|\Sigma|\cdot|G|)$. Interestingly, this result suggests that the difficulty in dealing with sets of generalized strings stems from the size of the set rather than from the length of the strings. For motifs given in the form of a single (generalized) string~$g$ along with a Hamming neighborhood bounded by a distance threshold~$d_{\max}$, we presented an algorithm that constructs a simple NFA in time~$\OO(\len{g}\cdot d_{\max})$. A third important class of motifs are position weight matrices (PWMs) with a score threshold~\cite{Sta84}. Such a motif could be transformed into a set of generalized strings, which in turn could be handled by the presented algorithm. Nonetheless, a more direct method to construct a simple NFA from a PWM is desirable and should be subject of future research.

In this article, we demonstrated that, for the considered pattern classes, a minimal DFA can be constructed directly, that is, without the intermediate step of a non-minimal DFA. A question we did not address, regards the size of the constructed minimal automata. In practice, we might still be faced with an exponential blow-up in the number of states. Thus, on the practical side, this study should be complemented by experiments measuring automata sizes and runtimes for typical motifs in future work.

\section{Acknowledgments}
I wish to thank Sven Rahmann for giving valuable comments on an earlier version of this manuscript, Wim Martens for pointing me to the paper by van Glabbeek and Ploeger~\cite{GlaPlo08}, Marcel Martin and Chris Schwiegelshohn for proof-reading, and Hoi-Ming Wong, who assisted in finding a graph layout with few edge crossings for Figure~\ref{fig:nfa_three_genstrings}.

\bibliographystyle{abbrv}
\bibliography{lit}

\end{document}